\documentclass{amsart}
\usepackage{amssymb, latexsym}

\newcommand{\beqa}{\begin{eqnarray*}}
\newcommand{\eeqa}{\end{eqnarray*}}
\newcommand{\beqn}{\begin{eqnarray}}
\newcommand{\eeqn}{\end{eqnarray}}

\newcommand{\R}{\mathbb R}

\newcommand{\Ha}{\mathbb H}
\newcommand{\D}{\mathbb D}

\newcommand{\mcA}{\mathcal A}

\newcommand{\f}{\frac}
\newcommand{\tf}{\tfrac}

\newcommand{\al}{\alpha}

\newcommand{\de}{\delta}
\newcommand{\la}{\lambda}

\newcommand{\om}{ \omega}
\newcommand{\Om}{ \Omega}

\newcounter{cnt1}
\newcounter{cnt2}
\newcounter{cnt3}
\newcommand{\blr}{\begin{list}{$($\roman{cnt1}$)$}
 {\usecounter{cnt1} \setlength{\topsep}{0pt}
 \setlength{\itemsep}{0pt}}}
\newcommand{\bla}{\begin{list}{$($\alph{cnt2}$)$}
 {\usecounter{cnt2} \setlength{\topsep}{0pt}
 \setlength{\itemsep}{0pt}}}
\newcommand{\bln}{\begin{list}{$($\arabic{cnt3}$)$}
 {\usecounter{cnt3} \setlength{\topsep}{0pt}
 \setlength{\itemsep}{0pt}}}
\newcommand{\el}{\end{list}}
\newtheorem{thm}{Theorem}
\newtheorem{lem}[thm]{Lemma}
\newtheorem{cor}[thm]{Corollary}

\newtheorem{Def}[thm]{Definition}

\newtheorem{rem}[thm]{Remark}
\newcommand{\Rem}{\begin{rem} \rm}
\newcommand{\bdfn}{\begin{Def} \rm}
\newcommand{\edfn}{\end{Def}}

\newcommand{\ba}{\begin{array}}
\newcommand{\ea}{\end{array}}

\date{}
\begin{document}
\title{\bf A sufficiency class for global (in time) solutions to the 
3D NavierÐStokes equations II}
\author[Gill]{T. L. Gill}
\address[Tepper L. Gill]{ Department of Electrical Engineering, Howard University\\
Washington DC 20059 \\ USA, {\it E-mail~:} {\tt tgill@howard.edu}}
%\and
\author[Zachary]{W. W. Zachary}
\address[Woodford W. Zachary]{ Department of Electrical Engineering \\ Howard
University\\ Washington DC 20059 \\ USA, {\it E-mail~:} {\tt
wwzachary@earthlink.net}}
% [-1ex] \normalsize \sc Howard University\\
%\abstract{ write abstract here}
\date{}
%thispagestyle{empty}
\subjclass{Primary (35Q30) Secondary(47H20), (76DO3) }
\keywords{Global (in time), 3D-Navier-Stokes Equations}
\begin{abstract} In this paper, we simplify and extend the results of \cite{GZ} to include the case in which $\Om =\R^3$.  Let ${[L^2({\mathbb{R}}^3)]^3 }$ be the Hilbert space of square integrable functions on ${\mathbb {R}}^3 $ and let ${\mathbb H}[{\mathbb{R}}^3]^3 =: {\mathbb H}$ be the completion of the set, $\left\{ {{\bf{u}} \in (\mathbb {C}_0^\infty  [ \R^3 ])^3 \left. {} \right|\,\nabla  \cdot {\bf{u}} = 0} \right\}$, with respect to the inner product of ${[L^2({\mathbb{R}}^3)]^3} $. In this paper, we consider sufficiency conditions on a class of functions in ${\mathbb H}$  which allow global-in-time strong solutions to the three-dimensional Navier-Stokes equations on  ${\mathbb {R}}^3 $. These equations describe the time evolution of the fluid velocity and pressure of an incompressible viscous homogeneous Newtonian fluid in terms of a given initial velocity and given external body forces.  Our approach uses the analytic nature of the Stokes semigroup to construct an equivalent norm for $\mathbb{H}$ which allows us to prove a reverse of the Poincar\'{e} inequality.  This result allows us to  provide strong bounds on the nonlinear term. We then prove that, under appropriate conditions, there exists a positive constant $ {{{u}}_ +}$, depending only on the domain, the viscosity and the body forces such that, for all functions in a dense set $\mathbb{D}$ contained in the closed ball ${{\mathbb B} ( {\mathbb {R}}^3 )}=: {\mathbb B}$ of radius $ (1/2){{u}_ +} $ in ${\mathbb {H}}$, the Navier-Stokes equations have unique strong solutions in ${\mathbb C}^{1} \left( {(0,\infty ),{\mathbb {H}}} \right)$.
\end{abstract}
\maketitle
\section*{Introduction} 
Let ${[L^2({\mathbb{R}}^3)]^3}$ be the Hilbert space of square integrable functions on ${\mathbb {R}}^3$ and let ${\mathbb {H}}_0[ {\mathbb {R}}^3 ]$ be the completion of the set of functions in $\left\{ {{\bf{u}} \in \mathbb {C}_0^\infty  [ {\mathbb {R}}^3 ]^3 \left. {} \right|\,\nabla  \cdot {\bf{u}} = 0} \right\}$ which vanish at infinity with respect to the inner product of ${[L^2({\mathbb{R}}^3)]^3 }$, and let ${\mathbb{V}}_0[ {\mathbb {R}}^3 ]$ be the completion of the above functions which vanish at infinity with respect to the inner product of $\mathbb{H}_0^1[ {\mathbb {R}}^3 ]$, the functions in ${\mathbb{H}}_0 [ {\mathbb {R}}^3 ]$ with weak derivatives in ${[L^2({\mathbb{R}}^3)]^3 }$.  The global-in-time classical Navier-Stokes initial-value problem (on $ \mathbb{R}^3 {\text{ and all }}T > 0$) is to find  functions ${\mathbf{u}}:[0,T] \times {\mathbb {R}}^3  \to \mathbb{R}^3$ and $p:[0,T] \times {\mathbb {R}}^3  \to \mathbb{R}$ such that
\beqn
\begin{gathered}
  \partial _t {\mathbf{u}} + ({\mathbf{u}} \cdot \nabla ){\mathbf{u}} - \nu \Delta {\mathbf{u}} + \nabla p = {\mathbf{f}}(t){\text{ in (}}0,T) \times {\mathbb {R}}^3 , \hfill \\
  {\text{                              }}\nabla  \cdot {\mathbf{u}} = 0{\text{ in (}}0,T) \times {\mathbb {R}}^3 {\text{ (in the weak sense),}} \hfill \\
  {\text{                              }}\mathop {\lim }\limits_{\left\| {\mathbf{x}} \right\| \to \infty } {\mathbf{u}}(t,{\mathbf{x}}) = 0{\text{ on }}\left( {0,T} \right) \times \mathbb{R}^3, \hfill \\
  {\text{                              }}{\mathbf{u}}(0,{\mathbf{x}}) = {\mathbf{u}}_0 ({\mathbf{x}}){\text{ in }}{\mathbb {R}}^3 . \hfill \\ 
\end{gathered} 
\eeqn
The equations describe the time evolution of the fluid velocity ${\mathbf{u}}({\mathbf{x}},t)$ and the pressure $p$ of an incompressible viscous homogeneous Newtonian fluid with constant viscosity coefficient $\nu $ in terms of a given initial velocity ${\mathbf{u}}_0 ({\mathbf{x}})$ and given external body forces
${\mathbf{f}}({\mathbf{x}},t)$.  (Note that our third condition, $\mathop {\lim }\limits_{\left\| {\mathbf{x}} \right\| \to \infty } {\mathbf{u}}(t,{\mathbf{x}}) = 0{\text{ on }}\left( {0,T} \right) \times \mathbb{R}^3$, is natural in this case since it is well-known that $\mathbb{H}_0^k [ {\mathbb {R}}^3 ]^3=\mathbb{H}^k [ {\mathbb {R}}^3 ]^3$ (see Stein \cite{S} or \cite{SY}.) 
\section*{Purpose}
Let $\mathbb{P}$ be the (Leray) orthogonal projection of 
$(L^2 [ {\mathbb {R}}^3 ])^3$ 
onto ${{\mathbb{H}}_0}[ {\mathbb {R}}^3]$ and define the Stokes operator by:  $ {\bf{Au}} = : -\mathbb{P} \Delta {\bf{u}}$, 
for ${\bf{u}} \in D({\bf{A}}) \subset {\mathbb{H}}_0^{2}[ {\mathbb {R}}^3]$, the domain of ${\bf{A}}$.      The purpose of this paper is to prove that there exists a number $ {{{u}}_ +} $, depending only on ${\bf{A}}$,  $f$ and  $\nu $ such that, for all functions in a certain subset (defined in the paper) of
$
\mathbb{D} = D({\bf{A}}) \cap \mathbb{B},
$
 where ${{\mathbb B}}$ is the closed ball of radius 
$ \tfrac{1}{2}{{u}}_ + $ in ${{\mathbb H}_0( {\mathbb {R}}^3 )}$, the Navier-Stokes equations have unique strong solutions in 
$
{\bf{u}}\in L_{\text{loc}}^\infty[[0,\infty); {\mathbb {V}}_0( {\mathbb {R}}^3)]
\cap \mathbb{C}^1[(0,\infty);{\mathbb {H}}_0( {\mathbb {R}}^3 )].$
\section*{Preliminaries}
Applying the Leray projection to equation (1), with 
${{C}}({\mathbf{u}},{\mathbf{u}}) = \mathbb{P}({\mathbf{u}} \cdot \nabla ){\mathbf{u}}$, we can recast equation (1) in the standard form:
\beqn
\begin{gathered}
  \partial _t {\mathbf{u}} =  - \nu {\mathbf{Au}} - {{C}}({\mathbf{u}},{\mathbf{u}}) + \mathbb{P}{\mathbf{f}}(t){\text{ in (}}0,T) \times \R^3 , \hfill \\
  {\text{                              }}\mathop {\lim }\limits_{\left\| {\mathbf{x}} \right\| \to \infty } {\mathbf{u}}(t,{\mathbf{x}}) = 0{\text{ on }}\left( {0,T} \right) \times \mathbb{R}^3, \hfill \\
  {\text{                              }}{\mathbf{u}}(0,{\mathbf{x}}) = {\mathbf{u}}_0 ({\mathbf{x}}){\text{ in }}\R^3, \hfill \\ 
\end{gathered} 
\eeqn
where we have used the fact that the orthogonal complement of ${\Ha}_0 $ relative to $\{{L}^{2}(\R^3)\}^3 $ is $\{ {\mathbf{v}}\,:\;{\mathbf{v}} = \nabla q,\;q \in (H^1)^3 \} $ to eliminate the pressure term (see Galdi [GA] or [SY, T1,T2]). 
\begin{Def}  We say that the operator ${\mathcal{A}}( \cdot ,t)$ is (for each $t$) 
\begin{enumerate}
\item
0-Dissipative if $
\left\langle {{\mathcal{A}}({\mathbf{u}},t),{\mathbf{u}}} \right\rangle _{\mathbb{H}}  \le 0$.
\item
Dissipative if 
$\left\langle {{\mathcal{A}}({\mathbf{u}},t) - {\mathcal{A}}({\mathbf{v}},t),{\mathbf{u}} - {\mathbf{v}}} \right\rangle _{\mathbb{H}}  \le 0$.
\item
Strongly dissipative if there  exists an $ \de > 0$ such that
$$
\left\langle {{\mathcal{A}}({\mathbf{u}},t) - {\mathcal{A}}({\mathbf{v}},t),{\mathbf{u}} - {\mathbf{v}}} \right\rangle _{\mathbb{H}}  \le  - \de \left\| {{\mathbf{u}} - {\mathbf{v}}} \right\|_{\mathbb{H}}^2. 
$$
\end{enumerate}
\end{Def}

Note that, if ${\mathcal{A}}( \cdot ,t)$ is a linear operator, definitions (1) and (2) coincide.  Theorem 2 below is essentially due to Browder \cite{B},  see Zeidler {\cite[Corollary 32.27, page 868 and Corollary 32.35 page 887, in Vol. IIB]{Z}}, while Theorem 3 is from Miyadera \cite[p. 185, Theorem 6.20]{M}, and is a modification of the Crandall-Liggett Theorem \cite{CL} (see the appendix to the first section of \cite{CL}) . 
\begin{thm} Let $\mathbb{B}$ be a closed, bounded, convex subset of $
\mathbb{H}$.  If ${\mathcal{A}}( \cdot ,t):\mathbb{B} \to \mathbb{H}$ is  a strongly dissipative mapping for each fixed $t \ge 0$, then for each ${\mathbf{b}} \in \mathbb{B}$, there is a 
${\mathbf{u}} \in \mathbb{B}$ with ${\mathcal{A}}({\mathbf{u}},t) = {\mathbf{b}}$ (i.e., the range, $
Ran{\text{[}}{\mathcal{A}}( \cdot ,t)] \supset \mathbb{B}$).
\end{thm}	
\begin{thm} Let ${\mathcal{A}( \cdot ,t)}, t \in I = [0,\infty ){\text{\} }}$ be a family of operators defined on $\mathbb{H}$ with domains $
D(\mathcal{A}( \cdot ,t)) = D$, independent of $t$. We assume that $\mathbb{D} = D(A) \cap \mathbb{B}$ is a closed convex set (in an appropriate topology):
\begin{enumerate}
\item
The operator $\mathcal{A}( \cdot ,t)$ is the generator of a contraction semigroup for each
$t \in I$.
\item
The function $\mathcal{A}({\mathbf{u}}, t)$ is continuous in both variables on $
 \mathbb{D} \times I $.
\end{enumerate}
Then, for every ${\mathbf{u}}_0  \in \mathbb{D}$, the problem 
$\partial _t {\mathbf{u}}(t,{\mathbf{x}}) = \mathcal{A}({\mathbf{u}}(t,{\mathbf{x}}), t)$, 
${\mathbf{u}}(0,{\mathbf{x}}) = {\mathbf{u}}_0 ({\mathbf{x}})$, has a unique solution 
${\mathbf{u}}(t,{\mathbf{x}}) \in \mathbb{C}^1 (I;\mathbb{D})$.
\end{thm} 
\subsection*{Stokes Equation}
The difficulty in proving the existence of global-in-time strong solutions for equation (2) is directly linked to the problem of getting good estimates for the nonlinear term ${{C}}({\mathbf{u}},{\mathbf{u}})$.  In \cite{GZ}, we obtained an extension of the important result due to Constantin and Foias \cite{CF}).  This result, see below, is one of the major estimates used to study this equation.    In what follows, we assume that ${\bf u},\, {\bf v} \in D({\bf{A}})$.
\begin{thm}   Let ${0 \le \alpha_i, \ 1 \le i \le 3}$, satisfy 
 ${ \alpha_1+\alpha_2+\alpha_3 = 3/2}$ and
\beqa
(\alpha _1 ,\alpha _2 ,\alpha _3 ) \notin \left\{ {(3/2,0,0),(0,3/2,0),(0,0,3/2)} \right\}.
\eeqa
Then there is a positive constant $c=c(\al_i, \mathbb{R}^3)$ such that
\beqn
\left| {\left\langle {{{C}}({\mathbf{u}},{\mathbf{v}}),{\mathbf{w}}} \right\rangle _\mathbb{H} } \right| \le c\left\| {{\mathbf{A}}^{\alpha _1 /2} {\mathbf{u}}} \right\|_\mathbb{H} \left\| {{\mathbf{A}}^{(1 + \alpha _2 )/2} {\mathbf{v}}} \right\|_\mathbb{H} \left\| {{\mathbf{A}}^{\alpha _3 /2} {\mathbf{w}}} \right\|_\mathbb{H}. 
\eeqn
\end{thm}
In \cite{GZ} we showed that, by renorming $\mathbb{H}$, we could prove a very strong inequality for equation (3).  Since this result is not well-known and important for this paper, we give a proof. First we need to review the Stokes equation.

If we drop the nonlinear term, we get the well-known Stokes equation ($\mathbb{P}{\mathbf{f}}(t)={\mathbf 0}$):
\beqa
\begin{gathered}
  \partial _t {\mathbf{u}} =  - \nu {\mathbf{Au}}  {\text{ in (}}0,T) \times \R^3 , \hfill \\
  {\text{                              }}\mathop {\lim }\limits_{\left\| {\mathbf{x}} \right\| \to \infty } {\mathbf{u}}(t,{\mathbf{x}}) = 0{\text{ on }}\left( {0,T} \right) \times \mathbb{R}^3, \hfill \\
  {\text{                              }}{\mathbf{u}}(0,{\mathbf{x}}) = {\mathbf{u}}_0 ({\mathbf{x}}){\text{ in }}\R^3. \hfill \\ 
\end{gathered} 
\eeqa
A proof of the next theorem may be found in Sell and You \cite{SY} (page 114):
\begin{thm}
Let  $\bf{A}$ be the Stokes operator on ${\mathbb{R}}^3$.  Then the following holds:
\begin{enumerate}
\item The operator $\bf{A}$  is a positive selfadjoint generator of a contraction semigroup $T(t)$.
\item The operator $\bf{A}$ is sectorial and $T(t)$ is analytic.
\end{enumerate}
\end{thm}
\subsection*{Equivalent Norms}
Recall that an equivalent norm, $\left\| {\, \cdot \,} \right\|_{\mathcal{H},1}$, on a Hilbert space $\mathcal{H}$, with norm $\left\| {\, \cdot \,} \right\|_{\mathcal{H}}$, is one that satisfies: (for positive constants $M,\; M_1$) 
\[
\left\| u \right\|_\mathcal{H}  \leqslant M \left\| u \right\|_{\mathcal{H},1}  \leqslant M_1 \left\| u \right\|_\mathcal{H},\; u \in \mathcal{H}. 
\]
It is easy to show that any equivalent norm on $\mathcal{H}$ can be identified with a transformation of $\mathcal{H}$ which preserves the topology. In order to see how  an equivalent norm can help us, let $T(t)=exp\{-t\bf{A}\}$ be the analytic contraction semigroup  generated by the Stokes operator $\bf{A}$, with $\left\| {T(t)\bf u} \right\|_\mathbb{H}  \leqslant  e^{-\omega t} \left\| \bf u \right\|_\mathbb{H}$.  Let $S(t)= e^{\om t}T(t)$ and choose $M$  so that $\left\| \bf u \right\|_{ {\mathbb{H}},1}  = \left\| {S(r)\bf u} \right\|_ {\mathbb{H}}$ is an equivalent norm, where $r$ is a {\it fixed value}, to be determined.  Since $\bf{A}$ is analytic,  there is a constant $c_z$ such that, for ${\bf u} \in D({\bf A}^z)$,
\[
\left\| {{\mathbf{A}}^z {\mathbf{u}}} \right\|_{\mathbb{H},1}  = e^{ \omega r} \left\| {{\mathbf{A}}^z T(r){\mathbf{u}}} \right\|_\mathbb{H}  \leqslant e^{ \omega r}  e^{-\omega r} \frac{c_z}
{{(r)^z }}\left\| {\mathbf{u}} \right\|_\mathbb{H}  \leqslant \frac{{Mc_z}}
{{(r)^z }}\left\| {\mathbf{u}} \right\|_{\mathbb{H},1}.
\]
Since the norms are equivalent, we also have that
\beqn
\left\| {{\mathbf{A}}^z {\mathbf{u}}} \right\|_{\mathbb{H}} \le M \left\| {{\mathbf{A}}^z {\mathbf{u}}} \right\|_{\mathbb{H},1}   \leqslant \frac{{M^2c_z}}
{{(r)^z }}\left\| {\mathbf{u}} \right\|_{\mathbb{H},1}.
\eeqn

From Theorem 4, we have the following result: 
\begin{thm}
Let ${\bf u} \in D({\bf{A}})$, set ${\bf S}=S(r)$ and renorm $\mathbb{H}$ so that $\left\| \bf u \right\|_{\mathbb{H},1}  = \left\| {{\mathbf{S}} \bf u} \right\|_\mathbb{H}$.  We define ${\bf{b}}({\mathbf{u}},{\mathbf{v}} ,{\mathbf{w}})_{\mathbb{H},1}  = \left\langle {{\mathbf{S}}C({\mathbf{u}},{\mathbf{v}}),{\mathbf{Sw}}} \right\rangle _\mathbb{H}$. Then:
\begin{enumerate}
\item
If we let  ${\alpha_1=0}$, ${ \alpha_2=1}$ and $\alpha_3=1/2$, there are positive constants $c=c(\al_i, \R^3),\; c_1$ and $c_2$  such that
\beqn
\begin{gathered}
\left| {\left\langle {{{C}}({\mathbf{u}},{\mathbf{v}}),{\mathbf{w}}} \right\rangle _{\mathbb{H},1} } \right| \le \frac{{M^4 cc_1 c_2 }}
{{r^{5/4}  }} \left\| {\bf{u}} \right\|_{\mathbb{H},1} \left\| {\bf{w}} \right\|_{\mathbb{H},1} \left\| {\bf{v}} \right\|_{\mathbb{H},1} \; {\rm and}  \hfill \\
\left| {\left\langle {{{C}}({\mathbf{v}},{\mathbf{u}}),{\mathbf{w}}} \right\rangle _{\mathbb{H},1} } \right| \le \frac{{M^4 cc_1 c_2 }}
{{r^{5/4}  }} \left\| {\bf{u}} \right\|_{\mathbb{H},1} \left\| {\bf{w}} \right\|_{\mathbb{H},1} \left\| {\bf{v}} \right\|_{\mathbb{H},1}.  \hfill \\
\end{gathered}
\eeqn
\item
\beqn
max\{ \left\| {{{C}}({\mathbf{u}},{\mathbf{v}})} \right\|_{\mathbb{H},1}, \ \left\| {{{C}}({\mathbf{v}},{\mathbf{u}})} \right\|_{\mathbb{H},1} \} \leqslant \frac{{M^4 cc_1 c_2 }}
{{r^{5/4}  }}\left\| {\mathbf{u}} \right\|_{\mathbb{H},1} \left\| {\mathbf{v}} \right\|_{\mathbb{H},1}. 
\eeqn
\end{enumerate}
\end{thm}
\begin{proof}

We prove the first equation of (5), the proof of the second is similar.   Set $S(r)={\bf{S}}$ and ${\bf{S}}^2{\bf{w}}={\bf{w}}_1$,  then we have:  
\[
{\bf{b}}( {\bf{u}},{\bf{v}} ,{\bf{w}})_{\mathbb{H},1}  = \left\langle {\bf{S}}C({\bf{u}},{\bf{v}}),{\bf{S}{\bf{w}}} \right\rangle _\mathbb{H} = {\bf{b}}({\bf{u}},{\bf{v}},{\bf{w}}_1 )_{\mathbb{H}}. 
\]
Using the selfadjoint property of ${\bf{A}}$, and integration by parts, we have
$$
{\bf{b}}({\bf{u}},{\bf{v}},{\bf{w}}_1)_{\mathbb{H}}   =  -  {\bf{b}}({\bf{u}},{\bf{w_{1}}},{\bf{v}})_{\mathbb{H}}.
$$
It follows that:
\[
\left| {\left\langle {{C}}({\bf{u}},{\bf{v}}),{\bf{w}} \right\rangle _{\mathbb{H},1}} \right| 
\le c\left\| {{\bf{A}}^{\alpha _1 /2} {\bf{u}}} \right\|_{\mathbb{H}} \left\| {{\bf{A}}^{ (1 + \alpha _2 )/2} {\bf{w}}_{1} }\right\|_{\mathbb{H}} 
\left\| {{\bf{A}}^{\alpha _3 /2} {\bf{v}}} \right\|_{\mathbb{H}}. 
\]
Setting $ \alpha _1 =  0$, $\alpha _2 = 1\; \alpha _3 = 1/2$ and, using equation (4), we have: 
\[
\begin{gathered}
  \left| {\left\langle {{{C}}({\mathbf{u}},{\mathbf{v}}),{\mathbf{w}}} \right\rangle _{{\mathbb{H}},1} } \right| \le c \left\| {\bf{u}} \right\|_{\mathbb{H}} \left\| {{\mathbf{A}} {\mathbf{w}}_1 } \right\|_\mathbb{H}  \left\| {\bf{A}}^{1/4}{\bf{v}} \right\|_\mathbb{H}\hfill \\
  {\text{                       }} \le  \frac{{M cc_1 c_2}}
{{r^{5/4}  }}\left\| {\bf{u}} \right\|_{\mathbb{H}} \left\| {\bf{w}} \right\|_\mathbb{H} \left\| {\bf{v}} \right\|_{\mathbb{H}}  \hfill \\
  {\text{                       }} \le \frac{{M^4 cc_1 c_2 }}
{{r^{5/4}  }} \left\| {\bf{u}} \right\|_{\mathbb{H},1} \left\| {\bf{w}} \right\|_{\mathbb{H},1} \left\| {\bf{v}} \right\|_{\mathbb{H},1}.  \hfill \\ 
\end{gathered} 
\]
The proof of (6) is clear.
\end{proof}
The following extension of the Poincar\'{e} inequality will also prove useful.
\begin{lem}  Let ${\bf A}^{1/2}$ generate an analytic contraction semigroup $S(t)$.  If $r > 0$ is any  fixed number, then there exists an $\al=\al(r) >0$ such that for any ${\bf u} \in D({\bf A}^{1/2})$,
\[
\alpha^{1/2} \left\| {\mathbf{u}} \right\|_{H,1}  \leqslant r^{1/2}\left\| {{\mathbf{A}}^{1/2} {\bf u}} \right\|_{H,1}.
\]
\end{lem}
\begin{proof}  First observe that
\[
\int_0^{r^{1/2}} {{\mathbf{A}}^{1/2}S(t)} {\mathbf{u}}dt = \int_0^{r^{1/2}} {\frac{d}
{{dt}}S(t)} {\mathbf{u}}dt = S({r^{1/2}}){\mathbf{u}} - {\mathbf{u}}.
\]
Now choose $\al^{1/2}$ so that $\left\| {S({r^{1/2}}){\mathbf{u}} - {\mathbf{u}}} \right\|_{H,1}  \ge \al^{1/2} \left\| {\mathbf{u}} \right\|_{H,1}$.  It follows that
\[
\al^{1/2} \left\| {\mathbf{u}} \right\|_{H,1}  \leqslant \int_0^{r^{1/2}} {\left\| {S(t){\mathbf{A}}^{1/2} {\bf u}} \right\|_{H,1} dt}  \leqslant {r^{1/2}}\left\| {{\mathbf{A}}^{1/2} {\bf u}} \right\|_{H,1} .
\]
\end{proof}
\section*{M-Dissipative Conditions} 
Let us assume that 
$
{\mathbf{f}}(t) \in L^\infty[[0,\infty); {\mathbb H}]
$
and is H\"{o}lder continuous in $t$, with $\left\| {{\mathbf{f}}(t) - {\mathbf{f}}(\tau )} \right\|_{\mathbb{H},1}  \le d\left| {t - \tau } \right|^\theta,{\text{ }}d > 0,{\text{ }}0 < \theta  < 1$.  We can now rewrite equation (2) in the form:
\beqn
\begin{gathered}
  \partial _t {\mathbf{u}} =  {\mathcal{A}} ({\mathbf{u}},t) {\text{ in (}}0,T) \times \R^3 , \hfill \\
  {\mathcal{A}}({\mathbf{u}},t) =  - \nu{\bf A}{\mathbf{u}} -  {{C}}({\mathbf{u}},{\mathbf{u}}) +  \mathbb{P}{\mathbf{f}}(t). \hfill \\ 
\end{gathered} 
\eeqn
We begin with a study of the operator $ {\mathcal{A}}( \cdot ,t)$, for fixed $t$, and seek conditions depending on ${\mathbf{A}},  {\text{ }}\nu ,{\text{ }}  {\text{ and }}{\mathbf{f}}(t)$ which guarantee that $ {\mathcal{A}}( \cdot ,t)$ is m-dissipative for each $t$.  Clearly $
 {\mathcal{A}}( \cdot ,t)$ is defined on $D({\bf{A}}) $ and, since $ \nu \mathbf{A} $ is a closed positive (m-accretive) operator, $ - \nu {(\mathbf{A})}$ generates a linear contraction semigroup. Thus, we need to ensure that $\nu  {\mathcal{A}}( \cdot ,t)$ will be m-dissipative for each $t$.  The following Lemma follows from the properties of ${\bf f}(t)$.
\begin{lem} For $t \in I=[0, \infty)$ and, for each fixed ${\mathbf{u}} \in D({\bf{A}})$, $
 {\mathcal{A}}({\mathbf{u}},t)$ is H\"{o}lder continuous, with $
\left\| { {\mathcal{A}}({\mathbf{u}},t) -  {\mathcal{A}}({\mathbf{u}},\tau )} \right\|_{\mathbb{H},1}  \le d\left| {t - \tau } \right|^\theta$, where  $d$ is the H\"{o}lder constant for the function ${\mathbf{f}}(t)$. 
\end{lem}

\section*{Main Results} 
\begin{thm} Let $f = \sup _{t \in {\mathbf{R}}^ +  } \left\| {\mathbb{P}{\mathbf{f}}(t)} \right\|_{\mathbb{H},1}  < \infty $, then there exists a positive constants ${{u}}_ +, \; {u}_-  $, depending only on $f$, ${\mathbf{A}}$ and $\nu $  such that, for all ${\mathbf{u}}$ with $
0 \le {u}_- \le \left\| {\mathbf{u}} \right\|_{\mathbb{H},1}  \le {{u}}_ +  $, $ {\mathcal{A}}( \cdot ,t)$ is strongly dissipative. 
\end{thm}
\begin{proof} The proof of our first assertion has two parts. First, we require that the nonlinear operator $ {\mathcal{A}}( \cdot ,t)$
 be 0-dissipative, which gives us an upper bound ${{u}}_ +  $ and lower bound ${{u}}_ -  $ in terms of the norm (i.e., $\left\| {\mathbf{u}} \right\|_{\mathbb{H},1}  \leqslant {{u}}_ + $ ).  We then use this part to show that $ {\mathcal{A}}( \cdot ,t)$ is strongly dissipative on any closed convex  ball, $ \mathbb{D} $ inside the annulus defined by$ \left\{ {{\mathbf{u}} \in D({\bf{A}}):0 \le {{u}}_ -  \le \left\| {\mathbf{u}} \right\|_{\mathbb{H},1}  \leqslant \tfrac{1}{2} {{u}}_ +  } \right\}$.  

Part 1) 
From equation (7), we consider the expression
\[
\begin{gathered}
  \left\langle {{\mathcal A}({\mathbf{u}},t),{\mathbf{u}}} \right\rangle _{\Ha,1}  =  - \nu \left\langle {{\mathbf{Au}},{\mathbf{u}}} \right\rangle _{\Ha,1}  + \left\langle {\left[ { - C({\mathbf{u}},{\mathbf{u}}) + \mathbb{P}{\mathbf{f}}} \right],{\mathbf{u}}} \right\rangle _{\Ha,1}  \hfill \\
   =  - \nu \left\| {{\mathbf{A}}^{1/2} {\mathbf{u}}} \right\|_{\Ha,1}^2  - \left\langle {C({\mathbf{u}},{\mathbf{u}}),{\mathbf{u}}} \right\rangle _{\Ha,1}  + \left\langle {\mathbb{P}{\mathbf{f}},{\mathbf{u}}} \right\rangle _{\Ha,1}.  \hfill \\ 
\end{gathered} 
\]
We now use the fact that 
\[
\tfrac{\alpha }
{r}\left\| {\mathbf{u}} \right\|_{\Ha,1}^2  \leqslant \left\| {{\mathbf{A}}^{1/2} {\mathbf{u}}} \right\|_{\Ha,1}^2  \Rightarrow  - \nu \tfrac{\alpha }
{r}\left\| {\mathbf{u}} \right\|_{\Ha,1}^2  \geqslant  - \nu \left\| {{\mathbf{A}}^{1/2} {\mathbf{u}}} \right\|_{\Ha,1}^2 
\]
to get that
\beqn
  \left\langle {{\mcA}({\mathbf{u}},t),{\mathbf{u}}} \right\rangle _{\mathbb{H},1}  \le   - \nu \tfrac{\alpha }
{r}\left\| {\mathbf{u}} \right\|_{\Ha,1}^2  +  \frac{{M^4 cc_1 c_2 }}
{{r^{5/4}  }}\left\| {{\mathbf{u}}} \right\|_{\mathbb{H},1}^3  +  f \left\| { {\mathbf{u}}} \right\|_{\mathbb{H},1}.  
\eeqn 
In the last line, we used our estimate from Theorem 6.

Since $\left\| {\mathbf{u}} \right\|_{\mathbb{H},1}  > 0$, we have that ${\mcA}( \cdot ,t)$ is 0-dissipative if
\[
 - \frac{{\nu \alpha }}
{r}\left\| {\mathbf{u}} \right\|_{H,1}  + \frac{{M^4 cc_1 c_2 }}
{{r^{5/4} }}\left\| {\mathbf{u}} \right\|_{H,1}^2  + f \leqslant 0
\]
Solving the equality, we get that
\beqa
{{u}}_ \pm   = \tfrac{\nu \al r^{1/4}}
{2M^4 c c_1c_2} \left\{ {1 \pm \sqrt {1 - ({{4fM^4cc_1c_2 )} \mathord{\left/
 {\vphantom {{4fM^4cc_1c_2)} {(\nu ^2 r^{1/2} \al^2 )}}} \right.
 \kern-\nulldelimiterspace} {(\nu ^2 r^{1/2} \al^2 )}}} } \right\} = \tfrac{\nu \al r^{1/4}}
{2M^4 c c_1c_2} \left\{ {1 \pm \sqrt {1 - \gamma } } \right\},
\eeqa
where $
\gamma  = ({{4r^{3/4}fM^4cc_1c_2 )} \mathord{\left/
 {\vphantom {{4r^{3/4}fM^4cc_1c_2 )} {(\nu ^2 \al^2  )}}} \right.
 \kern-\nulldelimiterspace} {(\nu ^2 \al^2 )}}.
$
Since we want real distinct solutions, we must require that 
\[
\gamma  = \frac{{4r^{3/4}fM^4cc_1c_2 }}
{{\nu ^2 \al^2}} < 1\quad  \Rightarrow \f{2M^2r^{3/8}}{\al}\left[ fcc_1c_2 \right]^{1/2}  < \nu. 
\]
It follows that, if $\mathbb{P}{\mathbf{f}} \ne {\mathbf{0}}$, then 
$
{{u}}_ -   < {{u}}_ + $ , and our requirement that $\mcA({\mathbf{u}},t)$ is 0-dissipative implies that, since our solution factors as 
$
(\left\| {\mathbf{u}} \right\|_{\mathbb{H},1}  - {{u}}_ +  )(\left\| {\mathbf{u}} \right\|_{\mathbb{H},1}  - {{u}}_ -  ) \le 0,
$
we must have that:
\beqa
\left\| {\mathbf{u}} \right\|_{\mathbb{H},1}  - {{u}}_ +   \le 0,{\text{  }}\left\| {\mathbf{u}} \right\|_{\mathbb{H},1}  - {{u}}_ -   \ge 0.
\eeqa
 It follows that, for  
${{u}}_ -   \le \left\| {\mathbf{u}} \right\|_{\mathbb{H},1}  \le {{u}}_ + $, 
$
\left\langle {\mcA({\mathbf{u}},t),{\mathbf{u}}} \right\rangle _{\mathbb{H},1}  \le 0$.  (It is clear that, when $
\mathbb{P}{\mathbf{f}}(t) = {\mathbf{0}}, \; {{u}}_ -   = {{0}}$, and ${{u}}_ +   = \tfrac{\nu \al r^{1/4}}
{M^4 c c_1c_2} $.)

Part 2): Now, for any ${\mathbf{u}},{\mathbf{v}} \in D({\mathbf{A}})$  with ${\mathbf{u}}-{\mathbf{v}} \in D({\mathbf{A}})$ and 
\[
u_- \le \min ({\text{ }}\left\| {\mathbf{u}} \right\|_{\mathbb{H},1} ,\left\| {\mathbf{v}} \right\|_{\mathbb{H},1} ) \le \max ({\text{ }}\left\| {\mathbf{u}} \right\|_{\mathbb{H},1} ,\left\| {\mathbf{v}} \right\|_{\mathbb{H},1} ) \le (1/2){{{u}}_ +},
\]
 we have that   
\beqa
\begin{gathered}
  \left\langle {{\mcA}({\mathbf{u}},t) - {\mcA}({\mathbf{v}},t),({\mathbf{u}} - {\mathbf{v}})} \right\rangle _{\mathbb{H},1}  =  -\nu \left\| {{\bf{A}}^{1/2} ({\mathbf{u}} - {\mathbf{v}})} \right\|_{\mathbb{H},1}^2  \hfill \\ 
  {\text{                                                    }} -  \left\langle { [{{C}}({\mathbf{u}},{\mathbf{u}} - {\mathbf{v}}) + {{C}}({\mathbf{v}}, {\mathbf{u-v}})],({\mathbf{u}} - {\mathbf{v}})} \right\rangle _{\mathbb{H},1}  \hfill \\
  {\text{                    }} \leqslant - \tfrac{\nu \al}{r} \left\| {{\mathbf{u}} - {\mathbf{v}}} \right\|_{\mathbb{H},1}^2 +  [1/(r^{5/4})]M^4 c{ c_1}c_2 \left\| {{\mathbf{u}} - {\mathbf{v}}} \right\|_{\mathbb{H},1}^2 \left( {\left\| {\mathbf{u}} \right\|_{\mathbb{H},1}  + \left\| {\mathbf{v}} \right\|_{\mathbb{H},1} } \right) \hfill \\
  {\text{                    }} \le - \tfrac{\nu \al}{r} \left\| {{\mathbf{u}} - {\mathbf{v}}} \right\|_{\mathbb{H},1}^2  + [1/(r^{5/4})]M^4 c{ c_1}c_2  \left\| {{\mathbf{u}} - {\mathbf{v}}} \right\|_{\mathbb{H},1}^2 {{u}}_ +   \hfill \\
  {\text{                    }} =  -  \tfrac{\nu \al}{r}  \left\| {{\mathbf{u}} - {\mathbf{v}}} \right\|_{\mathbb{H},1}^2  +  [1/(r^{5/4})]M^4 c{ c_1}c_2  \left\| {{\mathbf{u}} - {\mathbf{v}}} \right\|_{\mathbb{H},1}^2 \left( \tfrac{\nu \al r^{1/4}}{2M^4 c c_1c_2} \left\{ {1 + \sqrt {1 - \gamma } } \right\} \right) \hfill \\
  {\text{                    }} =  - \tfrac{\nu \al}{2r} \left\| {{\mathbf{u}} - {\mathbf{v}}} \right\|_{\mathbb{H},1}^2 \left\{ {1 - \sqrt {1 - \gamma } } \right\} \hfill \\
  {\text{                    }} =  - \de \left\| {{\mathbf{u}} - {\mathbf{v}}} \right\|_{\mathbb{H},1}^2 ,{\text{   }} \de = \tfrac{\nu \al}{2r}  \left\{ {1 - \sqrt {1 - \gamma } } \right\}. \hfill \\ 
\end{gathered} 
\eeqa
\end{proof} 
Let $\D$ be any closed convex set (in the graph norm of $\bf A$) inside the annulus bounded by $\tfrac{1}{2}{{u}}_+$ and ${{u}}_{-}$.
\begin{thm} The operator ${\mathcal{A}}(\cdot,t) $
 is closed, strongly dissipative and jointly continuous in ${\mathbf{u}}$ and $t$.  Furthermore, for each $t \in {\mathbf{R}}^ +  $ and $\beta  > 0$, 
$Ran[I - \beta  {\mathcal{A}}(t)] \supset \D$, so that $
 {\mathcal{A}}(t)$ is m-dissipative on $\D$. 
\end{thm}
\begin{proof} 
 It is easy to see that ${\mcA}( \cdot ,t)$ is closed. Since ${\mcA}( \cdot ,t)$ is strongly dissipative, it is maximal dissipative, so that $Ran[I - \beta  {\mathcal{A}}(\cdot, t)] \supset \D$.  It follows that ${\mcA}( \cdot ,t)$ is m-dissipative on $\D$ for each $t \in {\mathbf{R}}^ +  $ (since $\mathbb{H}$ is a Hilbert space).
To see that ${\mcA}( {\bf u} ,t)$ is continuous in both variables, let ${\mathbf{u}}_n ,{\mathbf{u}} \in \mathbb{B}$, $\left\| ({\mathbf{u}}_n  - {\mathbf{u}}) \right\|_{\mathbb{H},1}   \to 0$,
 with $t_n ,t \in I$ and $t_n  \to t$.  Then, if  $\left\| {{\mathbf{Au}}} \right\|_{\Ha,1}  \leqslant \tfrac{{Mc_3 }}{r}\left\| {\mathbf{u}} \right\|_{\Ha,1} $, we have
\beqa
\begin{gathered}
  \left\| { {\mathcal{A}}({\mathbf{u}}_n, t_n )  -  {\mathcal{A}}({\mathbf{u}}, t)} \right\|_{\mathbb{H},1}  \leqslant \left\| { {\mathcal{A}}({\mathbf{u}}, t_n ) -  {\mathcal{A}}({\mathbf{u}}, t) }\right\|_{\mathbb{H},1}   + \left\| { {\mathcal{A}}({\mathbf{u}}_n, t_n )  -  {\mathcal{A}}({\mathbf{u}}, t_n )} \right\|_{\mathbb{H},1}   \hfill \\
   = \left\| {{\text{[}}\mathbb{P}{\mathbf{f}}(t_n ) - \mathbb{P}{\mathbf{f}}(t)]} \right\|_{\mathbb{H},1}   + \left\| {\nu {\mathbf{A}}({\mathbf{u}}_n  - {\mathbf{u}}) + [{{C}}({\mathbf{u}}_n  - {\mathbf{u}},{\mathbf{u}}_n ) + {{C}}( {\mathbf{u}},{\mathbf{u}}_n  -{\mathbf{u}})]} \right\|_{\mathbb{H},1}   \hfill \\
   \leqslant d\left| {t_n  - t} \right|^\theta   + \nu \left\| {{\mathbf{A}}({\mathbf{u}}_n - {\mathbf{u}})} \right\|_{\mathbb{H},1}   + \left\| {{{C}}({\mathbf{u}}_n  - {\mathbf{u}},{\mathbf{u}}_n ) + {{C}}( {\mathbf{u}},{\mathbf{u}}_n  - {\mathbf{u}})} \right\|_{\mathbb{H},1}   \hfill \\
   \leqslant d\left| {t_n  - t} \right|^\theta   + \nu \tfrac{ Mc_3}{r} \left\| {({\mathbf{u}}_n  - {\mathbf{u}})} \right\|_{\mathbb{H},1}  +  \tfrac{ M^4cc_1c_2}{r^{5/4}} \left\| {({\mathbf{u}}_n  - {\mathbf{u}})} \right\|_{\mathbb{H},1} \left\{ {\left\| {{\mathbf{u}}_n } \right\|_{\mathbb{H},1}   + \left\| {{\mathbf{u}}} \right\|_{\mathbb{H},1}  } \right\} \hfill \\
   \leqslant d\left| {t_n  - t} \right|^\theta   + \nu \tfrac{M c_3}{r} \left\| {({\mathbf{u}}_n  - {\mathbf{u}})} \right\|_{\mathbb{H},1}   +  
\tfrac{M^4 cc_1c_2}{r^{5/4}} \left\| {({\mathbf{u}}_n  - {\mathbf{u}})} \right\|_{\mathbb{H},1}  {{u}}_+ . \hfill \\ 
\end{gathered} 
\eeqa
It follows that $ {\mathcal{A}}({\mathbf{u}}, t)$ is continuous in both variables.
\end{proof}
When $\mathbf{f}=\mathbf{0},\; \mathbb{D}$ is the graph closure of ${D}({\mathbf{A}}) \cap {\mathbb{B}}_+$ in the  $\mathbb H$ norm, where ${\mathbb{B}}_+$ is the ball of radius 
$\tf{1}{2} u_{+}$.  In this case, it follows that $\mathbb{D}$ is a closed, bounded, convex set.  We now have: 
\begin{thm} For each $T \in {\mathbf{R}}^ +$, $t \in (0,T)$ and ${\mathbf{u}}_0  \in \mathbb{D} $, the global-in-time Navier-Stokes initial-value problem in $\mathbb{R}^3 :$
\beqn
\begin{gathered}
  \partial _t {\mathbf{u}} + ({\mathbf{u}} \cdot \nabla ){\mathbf{u}} - \nu \Delta {\mathbf{u}} + \nabla p = {\mathbf{0}}{\text{ in (}}0,T) \times \mathbb{R}^3  , \hfill \\
  {\text{                              }}\nabla  \cdot {\mathbf{u}} = 0{\text{ in (}}0,T) \times \mathbb{R}^3  , \hfill \\
  {\text{                              }}\mathop {\lim }\limits_{\left\| {\mathbf{x}} \right\| \to \infty }{\mathbf{u}}(t,{\mathbf{x}}) = {\mathbf{0}}{\text{ on (}}0,T) \times \mathbb{R}^3  , \hfill \\
  {\text{                              }}{\mathbf{u}}(0,{\mathbf{x}}) = {\mathbf{u}}_0 ({\mathbf{x}}){\text{ in }}\mathbb{R}^3,  \hfill \\ 
\end{gathered} 
\eeqn
 has a unique strong solution ${\mathbf{u}}(t,{\mathbf{x}})$, which is in
 ${L_{\text{loc}}^2}[[0,\infty); {\mathbb {H}}]$ and in
$L_{\text{loc}}^\infty[[0,\infty); {\mathbb V}]
\cap \mathbb{C}^1[(0,\infty);{\mathbb H}]$.
\end{thm}
\begin{proof}
Theorem 3 allows us to conclude that, when ${\mathbf{u}}_0  \in \mathbb{D}$, the initial value problem is solved and the solution ${\mathbf{u}}(t,{\mathbf{x}})$ is in $\mathbb{C}^1[(0,\infty);{\mathbb D}]$.  Since $\mathbb{D} \subset \mathbb{H}^{2}$, it follows that ${\mathbf{u}}(t,{\mathbf{x}})$ is also in $\mathbb{ V}$ for each $t>0$.  It is now clear that, for any $T>0$,
\[
\int_0^T {\left\| {{\mathbf{u}}(t,{\mathbf{x}})} \right\|_{\mathbb{H}}^2 dt}  < \infty ,{\text{ and }}\sup _{0 < t < T} \left\| {{\mathbf{u}}(t,{\mathbf{x}})} \right\|_{\mathbb{V}}^2  < \infty .
\]
This gives our conclusion.
\end{proof}
When $\mathbf{f} \ne \mathbf{0}$, ${u}_{-} \ne {0}$.  Let $\Bbbk  = \left\{ {{\mathbf{u}}\;:\;\left\| {\mathbf{u}} \right\|_{{\mathbb{H}},1}  < {{u}}_ - \; \&, \; \left\| {\mathbf{u}} \right\|_{{\mathbb{H}},1} > \tf{1}{2}u_+  } \right\}$ and set $\mathbb{B}_ -   = \mathbb{B}   \cap \Bbbk ^c$, where $\Bbbk ^c$ is the complement of $\Bbbk $.  We can now take the graph closure of $\mathbb{B}{_{-}} \cap D({\bf A})$ and use the largest closed convex set containing the initial data inside this set.

\section*{Discussion}
 It is known that, if ${\mathbf{u}_0} \in \mathbb{V}$ and $\mathbf{f}(t) \in L^{\infty}[(0,\infty), \mathbb{H}]$, then there is a time $T> 0$ such that a weak solution with this data is uniquely determined on any subinterval of $[0,T)$ (see Sell and You, page 396, \cite{SY}).   Thus, we also have that: 
\begin{cor} For each $t \in {\mathbf{R}}^ + $ and $
{\mathbf{u}}_0  \in \mathbb{D} $ the Navier-Stokes initial-value problem on $ \mathbb{R}^3 :$
\beqn
\begin{gathered}
  \partial _t {\mathbf{u}} + ({\mathbf{u}} \cdot \nabla ){\mathbf{u}} - \nu \Delta {\mathbf{u}} + \nabla p = {\mathbf{f}}(t){\text{ in (}}0,T) \times \mathbb{R}^3  , \hfill \\
  {\text{                              }}\nabla  \cdot {\mathbf{u}} = 0{\text{ in (}}0,T) \times \mathbb{R}^3  , \hfill \\
  {\text{                              }}\mathop {\lim }\limits_{\left\| {\mathbf{x}} \right\| \to \infty }{\mathbf{u}}(t,{\mathbf{x}}) = {\mathbf{0}}{\text{ on (}}0,T) \times \mathbb{R}^3  , \hfill \\
  {\text{                              }}{\mathbf{u}}(0,{\mathbf{x}}) = {\mathbf{u}}_0 ({\mathbf{x}}){\text{ in }}\mathbb{R}^3,   \hfill \\ 
\end{gathered} 
\eeqn
 has a unique weak solution
${\mathbf{u}}(t,{\mathbf{x}})$ which is in
 ${L_{\text{loc}}^2}[[0,\infty); {\mathbb {H}^2}]$ and in
$L_{\text{loc}}^\infty[[0,\infty); {\mathbb V}]
\cap \mathbb{C}^1[(0,\infty);{\mathbb H}]$.
\end{cor}
As in \cite{GZ}, our results show that the Leray-Hopf weak solutions do not develop singularities if ${\mathbf{u}}_0 ({\mathbf{x}})\in {\mathbb {H}^2}$ (see Giga \cite{G} and references therein).  

We should note that the constant $\al$ in Lemma 7 depends on $r$ so we can't change $r$ without affecting $\al$.  This means that the size of $u_+$ need not increase with large values of $r$.

A close review of the results of this paper show that all theorems hold for the bounded domain case.  This provides an improvement of the results in \cite{GZ}. Furthermore, in that case, we can take $\al = \la_1$, the first eigenvalue of $\bf A$, which is independent of $r$.  This means that choosing larger values for $r$ could increase the possible size of $\D$ for bounded domains. However, the inequality for $\nu$ must be maintained, so that increasing $\D$ is not certain.

\end{document}